\documentclass[a4paper]{article}

\usepackage[utf8]{inputenc}

\usepackage{xcolor}
\usepackage{graphicx}

\usepackage{amsmath,amsthm,amssymb,latexsym}

\newtheorem{theorem}{Theorem}[section]

\newtheorem{lemma}[theorem]{Lemma}

\newtheorem{remark}[theorem]{Remark}

\newtheorem{example}[theorem]{Example}
\newtheorem{assumption}[theorem]{Assumption}

\newtheorem{claim}[theorem]{Claim}

\newcommand{\indicator}[1]{\mathbf{1}_{\{ #1 \}}}

\definecolor{myblue}{RGB}{0, 102, 204}

\usepackage{bbm}
\usepackage{enumitem}
\usepackage{xcolor}

\begin{document}

\title{On the utility problem in a market where price impact is transient}

\author{L\'or\'ant Nagy\thanks{HUN-REN Alfr\'ed R\'enyi Institute of Mathematics, Budapest, Hungary} \and Mikl\'os 
R\'asonyi\thanks{HUN-REN Alfr\'ed R\'enyi Institute of Mathematics and E\"otv\"os Lor\'and University, Budapest, 
Hungary; E-mail:\ rasonyi@renyi.hu}}


\date{\today}

\maketitle

\begin{abstract}
We consider a discrete-time model of a financial market where a risky asset is bought and sold with transactions having
a transient price impact. It is shown that the corresponding utility maximization problem admits a solution. We manage to remove
some unnatural restrictions on the market depth and resilience processes that were present in earlier work. A non-standard 
feature of the problem is that the set of attainable portfolio values may fail the convexity property.
\end{abstract}

\noindent Communicated by Dylan Possama\"\i.{}

\smallskip{}




\section{Introduction}

Investors' actions move market prices and make large position changes costly. More or less realistic models for
this \emph{price impact} phenomenon have been worked out in the mathematical finance literature.
One of the principal questions is the study of optimal investment in the presence of these (nonlinear)
transaction costs.

Price impact may be assumed \emph{instantaneous} if it affects the investor only at the moment of his/her portfolio rebalancing. 
This assumption leads to a relatively simple market dynamics,
see \cite{gr,pym}. At the other extreme, price impact may be \emph{permanent}, in which case
the investor's action pushes the price in a direction and this effect pertains to the whole future.
The most accurate description of reality is probably in between: price impact should
be \emph{transient}, with a gradually fading effect for the future, see \cite{peter-yan,bank-voss}.
The speed at which these effects disappear is characterized by \emph{market resilience}: if $r$ is
resilience then the bid-ask (half-)spread is diminished by a factor of $e^{-r}$ in one unit of time. The
size of the disturbance
caused by trading a unit amount of the asset is described by \emph{market depth}: if
$\delta$ is market depth then $1/\delta$ is the effect on the bid-ask spread of trading one unit of the risky asset. 
  
In the present article we prove that the utility maximization problem in discrete time for an agent
trading with transient price impact is well-posed: it admits a solution.
The problem under consideration has a complex, non-linear dynamics involving
all previous strategies at a given time. Moreover, the domain of optimization 
is \emph{non-convex}, which is a highly unusual feature. We rely on techniques for non-concave utilities, see \cite{laurence}.

A continuous-time model with \emph{instantaneous} price impact was considered in \cite{gr}: 
they proved in their Theorem 5.1 that the utility maximization problem (with a concave utility function) admits
a solution under mild conditions. 

In \cite{bank-voss} a more advanced model with \emph{transient} 
price impact was treated where the markets' resilience and depth were both assumed constant.
Theorem 3.3 of \cite{bank-voss} proved the existence of an optimal investment strategy in such a setting.

In \cite{peter-yan} market resilience and depth were both allowed to be stochastic processes
but a related monotonicity condition was imposed in their Assumption 2.4 which implies
that the set of attainable portfolio values is convex. That condition holds for
constant resilience and depth but otherwise it is rather restrictive. The paper \cite{peter-yan} did not
provide a general existence theorem for optimizers but a superreplication result (Theorem 3.2 in \cite{peter-yan})
with a dual characterization of contingent claims that can be superhedged from a given initial
position. They also proved a verification result (Corollary 3.5 in \cite{peter-yan}): a sufficient condition
implying that a given strategy is optimal. Transient price impact has been further studied in \cite{3}.

In the present work we are dealing with a discrete-time version of the model of
\cite{peter-yan}. Our purpose is to prove the existence of an optimal strategy for
a utility maximizer while removing the stringent monotonicity assumption of \cite{peter-yan}
on market resilience and market depth, see Theorem \ref{main} and Remark \ref{megjegyzes} below. 

Section \ref{setup} presents the technical details of our model, the main result (Theorem \ref{main}) and some discussions
about the lack of convexity and its implications. Proofs will 
then be provided starting with Section \ref{poof}. Section \ref{single} deals with the one-step case while Section \ref{dyna} 
provides the dynamic programming argument. 
Further reflections are given in Section \ref{conclu}.

\section{Setup and Results}\label{setup}

Fix a probability space $(\Omega,\mathcal{F},\mathbb{P})$ once and for all, together with a discrete-time
filtration $\mathcal{F}_{t}$, $t=0,\ldots,T$. All sigma-algebras in this paper are assumed to be
$\mathbb{P}$-complete and $\mathcal{F}_{0}$ is assumed to be trivial. This is necessary for applying measurable selection arguments as in Lemma
\ref{lemma::kettes} or Lemma \ref{gast}.
 
Mathematical expectation with respect to $\mathbb{P}$ will be denoted by ${E}[\,\cdot\,]$,{}
${E}_{t}[\, \cdot\,]$ stands for $\mathcal{F}_{t}$-conditional expectation. 
$\mathbf{1}_{A}$ denotes the indicator function of a set $A$.
$L^{0}$ is the set of all
real-valued random variables. 

We now present a discrete-time version of the model in \cite{peter-yan}.
The time horizon of the investor will be some $T\in\mathbb{N}$. In the $T$th step the investor
must liquidate her position in the risky asset hence genuine decisions are made only at 
the previous times $t=1,\ldots, T-1$. To have a nontrivial problem we thus need to assume $T\geq 2$.

The risky asset's midprice (that is, the middle point of the bid-ask spread)
is described by an adapted 
real-valued process $P_{t}$, $t=0,\ldots,T$.{}

We assume that the trader we consider can buy and sell this risky asset and she also
has a cash account with zero interest rate.
Position in the risky asset at time $t$ is denoted by $X_{t}$, $t=0,\ldots,T$, we assume $X_{0}:=0$.
At each time $t$ the investor makes a portfolio adjustment based on previous information (before the new price $P_{t}$
is revealed) hence $X_{t}$ is assumed $\mathcal{F}_{t-1}$ measurable, that is, the strategy process
$X_{t}$, $1\leq t\leq T$ is predictable. We follow the convention $X_{-1}:=0$.
The set of all strategies is denoted by $\mathcal{A}$. We define
$$
\mathcal{A}_{0}:=\{X\in\mathcal{A}:X_{T}=0\},
$$
the set of strategies that liquidate the position in the risky asset by the end of the time horizon.
We note here, that due to the dynamics utilized in the paper, presented below, maximization of the utility of the terminal wealth $\xi_{T}^{X}$ in $X$ is economically meaningful only
over the set of strategies $\mathcal{A}_{0}$. Outside of $\mathcal{A}_0$, an investor could attain a position 
in the bank account with favourable expected utility while having large negative positions in the risky asset, 
and we would need to deal with liquidation value: such scenarios will be excluded.


We model liquidity with two primitives, resilience rate, and market depth: market resilience is described by a non-negative 
adapted process denoted by $r_{t}\geq 0$, $t=0,\ldots,T-1$, and market depth is a positive adapted 
process $\delta_{t}>0$, $t=1,\ldots,T$. 

These two processes will determine the \emph{half-spread}, that is, the distance of the mid-price
from the best bid/ask price at the given moment. Because of this interpretation this process must
start from a non-negative initial value $\zeta_{0}^{X}:=\zeta_{0}\geq 0$ and then it follows a linear dynamics, namely
\begin{equation}\label{spred}
\zeta^{X}_{t+1}=e^{-r_{t}}\zeta^{X}_{t}+\frac{1}{\delta_{t+1}}|X_{t+1}-X_{t}|,\ 0\leq t\leq T-1.
\end{equation}
The investor's buying or selling actions $X_{t+1}-X_{t}$ widen the spread. The size of this
effect is determined by $\delta_{t+1}$. The dynamics \eqref{spred} implies that such effects diminish in time at the 
exponential rate $r_{t}$.

Finally, the cash account at time $t=1,\ldots,T$, considering that the investor pays the spread when trading, is calculated as
\begin{equation}\label{baccount}
\xi^{X}_{t+1}-\xi_{t}^{X}=-P_{t+1}(X_{t+1}-X_{t})-\zeta^{X}_{t+1}|X_{t+1}-X_{t}|,\ 0\leq t\leq T-1.
\end{equation}
Setting $\xi^{X}_{0}:=0$ for simplicity, with initial capital $z\in\mathbb{R}$, the investor has a 
wealth of $z+\xi_{T}^{X}$ at time $T$. The above model is a discrete-time counterpart of 
the model introduced in \cite{peter-yan} except that, in our setup, only transient impact is modeled while
\cite{peter-yan} allows a permanent component, too.

\begin{remark}{\rm The portfolio value dynamics \eqref{baccount} deviates from the usual way self-financing is understood.
Rebalancing of the portfolio should normally happen at the mid-price $P_{t}$ while here $P_{t+1}$ features in the formula.
This change of index represents an unavoidable delay between information (on which the investor's decision is based) and
actual portfolio rebalancing. Note that, in this model, the investor acts using market orders.
We assume $X$ to be predictable, that is, $X_{t+1}-X_{t}$ is $\mathcal{F}_{t}$-measurable:
decisions are based on the market data available up to time $t$ while at the actual moment of rebalancing, $t+1$,
new conditions hold: the new half-spread $\zeta^{X}_{t+1}$ and eventually our new cash account value $\xi^{X}_{t+1}$
are determined by these new conditions (namely, by $P_{t+1},\delta_{t+1}$) and by the action $X_{t+1}-X_{t}$ that we took.

If $P_{t+1}$ is replaced by $P_{t}$ in \eqref{baccount} then Theorem \ref{main} continues
to hold true with a somewhat different (and simpler) proof.}
\end{remark}

We further assume that the investor may possibly have certain
payment obligations during the period considered,  
described by an $\mathcal{F}_{T}$-measurable $\mathbb{R}$-valued
random variable $B$. Negative $B$ means that the investor \emph{receives} a random endowment during the trading period.

The investor aims to maximize her expected utility from terminal wealth, hence we fix a \emph{utility function} $u:\mathbb{R}\to\mathbb{R}$.

\begin{assumption}\label{uint} The function $u$ is non-decreasing, continuous, $\lim_{x \to -\infty}u(x)=-\infty$, and $u$ is bounded from above.
We furthermore assume that for all $x,y,z\in\mathbb{R}$ and for all $t=1,\ldots,T-1$, 
\begin{equation}\label{inti}
{E}[u(x+y P_{t}+zP_{T}-B)]>-\infty
\end{equation}
holds.
\end{assumption}

\begin{remark}{\rm If $u$ is concave then \eqref{inti} is implied by the simpler condition
\begin{equation*}
{E}[u(x+y P_{t}-B)]>-\infty,\ t=1,\ldots,T,\ x,y\in\mathbb{R}.
\end{equation*}
Indeed, by concavity of the mapping $v\to u(x+v-B)$,
$$
{E}\left[u(x+y P_{t}+zP_{T}-B)\right]\geq \frac{{E}\left[u(x+2y P_{t}-B)\right]+{E}\left[u(x+2zP_{T}-B)\right]}{2}.
$$
}
\end{remark}

Our next hypothesis stipulates that market depth is always above a fixed threshold.

\begin{assumption}\label{deltacska}
There is a constant $\delta>0$ such that $\delta_{t}\geq \delta$ almost surely, for all $t=1,\ldots,T$.	
\end{assumption}

\begin{remark}{\rm This assumption is rather mild. It is intuitively clear that, when $\delta_{t}\to 0$ for all $t$,
trading becomes too expensive and one does not trade at all.}
\end{remark}

Our main result asserts that an investor with an arbitrary initial capital $z\in\mathbb{R}$ may
find an optimal portfolio strategy $X^{*}(z)$.

\begin{theorem}\label{main} Let Assumption \ref{uint} and Assumption \ref{deltacska} be in force.
Then, for each $z\in\mathbb{R}$ there exists $X^{*}(z)\in \mathcal{A}_{0}$ such that
$$
\bar{u}(z):={E}\left[u\left(z+\xi^{X^{*}(z)}_{T}-B\right)\right]=\sup_{X\in\mathcal{A}_{0}}{E}\left[u\left(z+\xi^{X}_{T}-
B\right)\right].
$$
\end{theorem}

Theorem \ref{main} shows that, despite the possible lack of convexity for the set of attainable values
(see Example \ref{notconvex} below), the utility maximization
problem admits an optimal strategy. We will present the proof of Theorem \ref{main} 
in the coming sections, using
a customized dynamic programming procedure.

From now on, for any stochastic process $Z_{t}$, we denote its increments by 
$\Delta Z_{t}:=Z_{t}-Z_{t-1}$, $0\leq t\leq T$ with the convention $Z_{-1}:=0$. 
It is convenient to use another parametrization for strategies: for a given real-valued
process $H_{t}$, $1\leq t\leq T$ such that $H_{t}$ is $\mathcal{F}_{t-1}$-measurable,
we may assign a unique strategy $X_{t}$ such that $X_{0}=0$, $\Delta X_{t}=H_{t}$, 
$1\leq t\leq T$. For such strategies we will also use the alternative notations 
$\zeta^{H},\xi^{H}$ for the corresponding half-spread and portfolio value processes.
With a slight abuse of notation we will also write $H\in \mathcal{A}_{0}$ when
we really mean that the corresponding $X$ is in $\mathcal{A}_{0}$. Note that $H\in\mathcal{A}_{0}$
implies that $H_{T}=-\sum_{j=1}^{T-1}H_{j}$, in particular, $H_{T}$ is $\mathcal{F}_{T-2}$-measurable.

Introduce the notation $\rho_{j,t}:=\exp\left[ -\sum_{i=j}^{t-1}r_{i}\right]$, $1\leq t\leq T$,
$0\leq j\leq t$. Note that $\rho_{t,t}=1$.
From \eqref{spred} and \eqref{baccount} we derive the explicit formula
\begin{equation}\label{expli}
\xi_{T}^{X}=-\sum_{t=1}^{T} P_{t}H_{t}
-\sum_{t=1}^{T} |H_{t}|\left( \rho_{0,t}\zeta_{0} +\sum_{j=1}^{t} 
\frac{\rho_{j,t}}{\delta_{j}}|H_{j}|\right).	
\end{equation}

\begin{example}\label{notconvex}{\rm{} Let $T=3$, $r_{t}=0$ for all $0\leq t\leq 3$, $P_{0}=\zeta_{0}=B=z=0$ and let 
$P_{1},P_{2},P_{3}$ be independent standard Gaussian.
Let $\mathcal{F}_{t}$, $0\leq t\leq 3$ be the natural filtration of the process $P$.
Let $\delta_{1}=1$, $\delta_{2}=\delta_{3}=10$. We claim that the set
$$
\mathcal{S}:=\{Y\in L^{0}:\exists X\in\mathcal{A}_{0}\mbox{ such that }\xi^{X}_{T}\geq Y\mbox{ a.s.}\}
$$
fails convexity. We argue by contradiction. 
Convexity would imply, in particular, that for arbitrary \emph{non-negative deterministic} 
strategies $H_{i},G_{i}$,
$i=1,2$ such that $H_{3}:=-H_{1}-H_{2}$, $G_{3}=-G_{1}-G_{2}$, there would exist some $\bar{H}_{i}$, $i=1,2,3$ such that
$$
\xi^{\bar{H}}_{T}\geq \frac{\xi^{H}_{T}+\xi^{G}_{T}}{2}
$$
almost surely. In view of \eqref{expli}, this inequality can be rewritten as
\begin{eqnarray}\nonumber
& & \frac{1}{2}\left[H_{1}^{2}+H_{2}H_{1}+\frac{1}{10}H_{2}^{2}+(H_{1}+H_{2})H_{1}+\frac{1}{10}(H_{1}+H_{2})H_{2}+
\frac{1}{10}(H_{1}+H_{2})^{2}
\right]\\ 
\nonumber &+& \frac{1}{2}\left[G_{1}^{2}+G_{2}G_{1}+\frac{1}{10}G_{2}^{2}+(G_{1}+G_{2})G_{1}+\frac{1}{10}(G_{1}+G_{2})G_{2}
+\frac{1}{10}(G_{1}+G_{2})^{2}\right]\\
\nonumber &-& \left[\bar{H}_{1}^{2}+\bar{H}_{2}\bar{H}_{1}+\frac{1}{10}\bar{H}_{2}^{2}+(\bar{H}_{1}+\bar{H}_{2})
\bar{H}_{1}+\frac{1}{10}(\bar{H}_{1}+\bar{H}_{2})\bar{H}_{2}+
\frac{1}{10}(\bar{H}_{1}+\bar{H}_{2})^{2}\right] \\
\nonumber &\geq&	{}
\left(\bar{H}_{3}-\frac{H_{3}}{2}-\frac{G_{3}}{2}\right)P_{3}+\left(\bar{H}_{2}-\frac{H_{2}}{2}-\frac{G_{2}}{2}\right)P_{2}+
\left(\bar{H}_{1}-\frac{H_{1}}{2}-\frac{G_{1}}{2}\right)P_{1}.
\end{eqnarray}
Notice that the $\mathcal{F}_{2}$-conditional law of $(\bar{H}_{3}-H_{3}/2-G_{3}/2)P_{3}$
is nondegenerate Gaussian on the set $\{\bar{H}_{3}-H_{3}/2-G_{3}/2\neq 0\}$. Also, the left-hand
side and the last two terms of the right-hand side are $\mathcal{F}_{2}$ measurable. Hence
the above inequality necessarily implies that $\mathbb{P}(\bar{H}_{3}-H_{3}/2-G_{3}/2\neq 0)=0$.
By an analogous argument, also $\bar{H}_{i}-H_{i}/2-G_{i}/2=0$, $i=1,2$. But then 
\begin{eqnarray*}
& & \frac{1}{2}\left[\frac{21}{10}H_{1}^{2}+\frac{3}{10}H_{2}^{2}+\frac{23}{10}H_{1}H_{2}+
\frac{21}{10}G_{1}^{2}+\frac{3}{10}G_{2}^{2}+\frac{23}{10}G_{1}G_{2}\right]\\ 
&\geq& 
\frac{21}{10}\bar{H}_{1}^{2}+\frac{3}{10}\bar{H}_{2}^{2}+\frac{23}{10}\bar{H}_{1}\bar{H}_{2}.
\end{eqnarray*}
The latter property, however, badly fails. Take, for instance, $H_{1}=H_{2}=1$ and $G_{1}=1.5$, $G_{2}=0$.
We conclude that $\mathcal{S}$ is \emph{not} a convex set.}
\end{example}

\begin{remark}\label{megjegyzes}{\rm The paper \cite{peter-yan} made a monotonicity
assumption. In the current discrete-time context it would require that the process $\rho_{0,t}^{2}\delta_{t}$ is a.s.
strictly decreasing in $t$. Imagine, for instance, that $r_{t}=\exp(Q_{t})$, 
$\delta_{t}=\delta+\exp(W_{t})$ where $(Q_{t},W_{t})$ are discretizations
of a $2$-dimensional diffusion with non-degenerate covariance matrix. In this case
such a monotonicity property cannot hold: it may well happen with positive probability that $\delta_{t}$ 
moves upwards and $\rho_{0,t}^{2}$ downwards in such a manner that their product $\rho_{0,t}^{2}\delta_{t}$ increases in $t$.
This shows that even for rather standard stochastic models for $r_{t},\delta_{t}$, monotonicity is a severely restrictive condition.

Under their monotonicity hypothesis, \cite{peter-yan} showed that $\mathcal{S}$ is convex. The novel
contribution of our work is to drop such monotonicity assumptions, to show an example where $\mathcal{S}$
is non-convex and nevertheless to prove the existence 
of optimal strategies. Lack of convexity forces us to abandon standard techniques such as the Koml\'os theorem.}
\end{remark}


Now we point out in a simple example why
non-convexity of the domain of optimization may cause trouble in optimal investment problems.

\begin{example}\label{noncon}
{\rm We consider a one-step frictionless market model. 
Let the family of permitted strategies be $\Phi:=\{0,1\}$. That is, the investor
may take either a unit position in the risky asset or no position at all. Let us consider the utility function 
$u(x):=\ln(x)$ for $x>0$. 
Let the investor have initial capital $z>1$. Assume that the return on his investment is
a random variable $R$ with  $\mathbb{P}(R=2)=1/2=\mathbb{P}(R=-1)$. His indirect utility is then 
\begin{equation}\label{baru}
\bar{u}(z):=\sup_{\phi\in\Phi} {E}[u(z+\phi R)]=\max\left\{\ln(z),{}
\frac{\ln(z+2)+\ln(z-1)}{2}\right\},\ z>1.
\end{equation}
This function fails concavity: it is non-differentiable
at $z=2$ with the right-hand derivative being strictly larger than the left-hand derivative.

We conclude that even if the investor's utility is risk-averse (concave), his \emph{indirect} utility
may well fail this property when the family of permitted strategies is non-convex. Hence
in related multistep optimization problems
one needs to deal with a dynamic programming procedure involving \emph{non-concave} 
functions, as in \cite{laurence}.}
\end{example}

\begin{example}\label{numex}
{\rm We consider the transient impact model presented in Section \ref{setup} with $T=4$, 
$B=\zeta_{0}=P_{0}=0$, $r_{t}=0.01$, $0\leq t\leq 3$, and
$$
(\delta_{1},\delta_{2},\delta_{3},\delta_{4})=(10,10,15,5).
$$
Let $P_{t}=P_{t-1}+\varepsilon_{t}$, where the random variables $\varepsilon_{t}$ are independent and
$$
\mathbb{P}(\varepsilon_{t}=1)=0.77,\qquad
\mathbb{P}(\varepsilon_{t}=-1)=0.23.
$$
We use the non-concave utility function defined by
$$
u(x):=x,\qquad x\leq 0,
$$
and
$$
u(x):=\frac{\sigma(x-5)-\sigma(-5)}{1-\sigma(-5)},\qquad x>0,
$$
where
$$
\sigma(x):=\frac{1}{1+e^{-x}}.
$$
The dynamic programming procedure is evaluated numerically by restricting $H_{1},H_{2},H_{3}$ to an equidistant 
grid on $[-4,4]$ and setting $H_{4}=-H_{1}-H_{2}-H_{3}$. (This choice of grid seemed sufficient
for demonstrating the main features of the problem at hand.) Figure \ref{numfig} shows the resulting indirect utility 
together with $u$. Note that this indirect utility is denoted $\bar{u}$ in \eqref{baru}
but it equals the function $\tilde{G}_{0}$ of Section \ref{dyna} so we use the latter notation.

The numerical results suggest that the backward recursion has a concavifying effect on the 
successive value functions, nevertheless, in this example the resulting initial value 
function $\tilde{G}_{0}$ remains non-concave.}
\end{example}

\begin{figure}[h]

\centering

\includegraphics[width=0.85\textwidth]{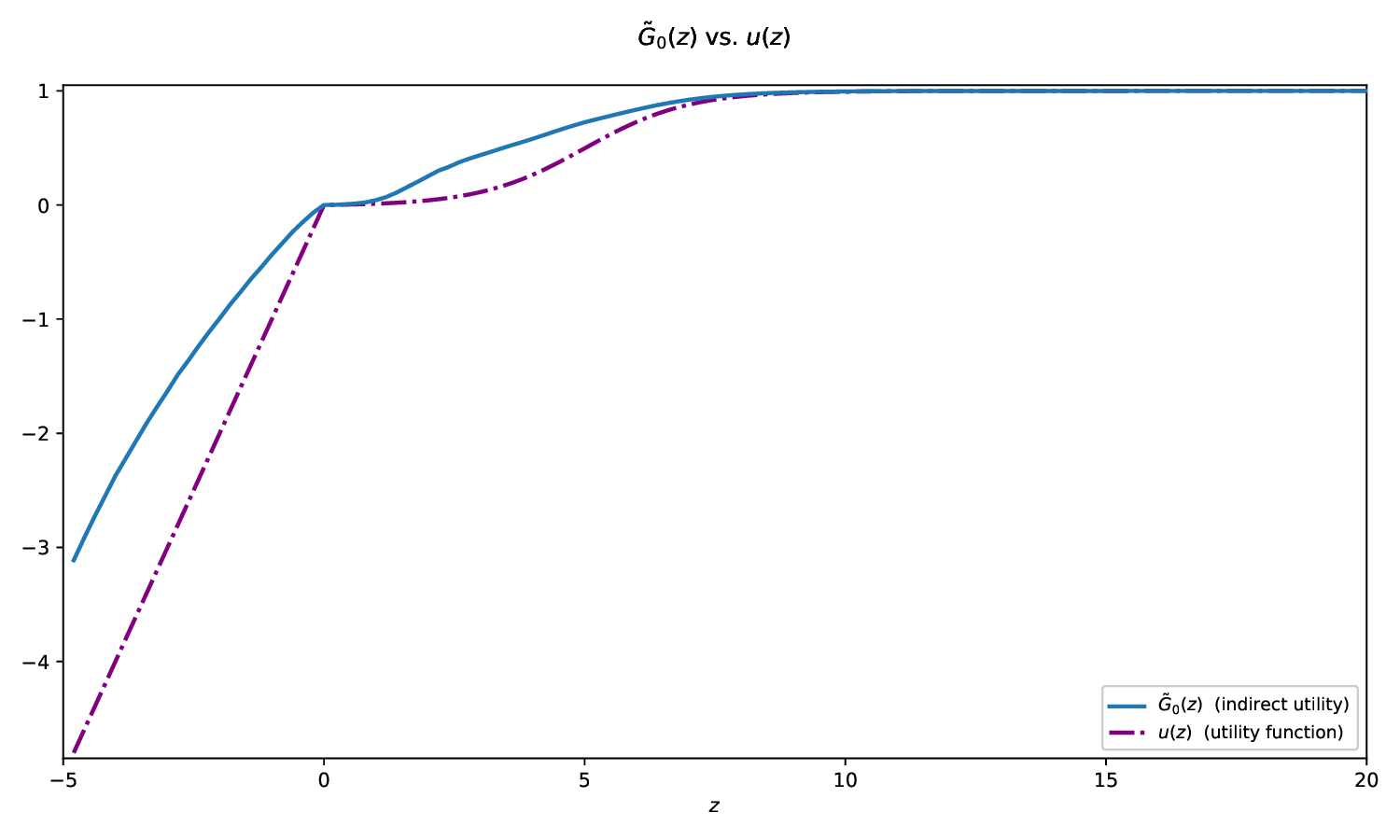}

\caption{The utility function $u$, and the indirect utility $\tilde{G}^*_{0}$ of Example \ref{numex}.}

\label{numfig}

\end{figure}

\section{Preparation for the Proof}\label{poof}



Recall from \eqref{baccount} and \eqref{spred} that
$$
\Delta \xi_{t}^{H}=-P_{t} H_{t}-|H_{t}|
\left(\rho_{0,t}\zeta_{0}
+\sum_{j=1}^{t}  
\frac{\rho_{j,t}}{\delta_{j}}|H_{j}|\right).
$$
Inspired by this formula, for $1\leq t\leq T$ and $h = (h_{1},\ldots,h_{t})$ we introduce the random functions
\begin{equation}\label{kappafg}
\kappa_{t}(\omega,h):= 
-P_{t}(\omega) h_{t}-|h_{t}|\left( \rho_{0,t}(\omega)\zeta_{0}+\sum_{j=1}^{t} 
\frac{\rho_{j,t}(\omega)}{\delta_{j}(\omega)} |h_{j}|\right)
\end{equation}
for all $h_{1},\ldots,h_{t}\in\mathbb{R}$. Note that the mapping, describing the innovation corresponding to a deterministic action of the trader,
$$
h_{t}\to \kappa_{t}(\omega,(h_{1},\ldots,h_{t-1},h_{t})),
$$
is \emph{concave} for every fixed $(h_{1},\ldots,h_{t-1}),\omega$, 
but $\kappa_{t}(\omega,\cdot)$, as a function of $t$ variables, has no reason to be concave.
Note also that innovation has an ``action-independent'' market bound, namely 
\begin{equation}\label{lambda}
\kappa_{t}(\omega,(h_{1},\ldots,h_{t}))\leq \lambda_{t}(\omega,h_{t}):=-P_{t}(\omega) h_{t}
-\frac{h_{t}^{2}}{\delta_{t}(\omega)} \leq \frac{P_t^2(\omega) \delta_t(\omega)}{4}.
\end{equation}

We recall Lemma 12 of \cite{laurence}. 

\begin{lemma}
\label{gast}
Let $(\Omega, {\cal H},P)$ be a (complete) probability space. Let $\Xi^d$ be the set of 
$\mathcal{H}$-measurable $d$-dimensional random variables.
Let $F:\Omega\times\mathbb{R}^d\to\mathbb{R}$ be a function such that for almost all $\omega\in\Omega$,
$F(\omega,\cdot)$ is continuous and for each $y\in\mathbb{R}^d$, $F(\cdot,y)$ is ${\cal H}$-measurable. Let $K>0$ be
an ${\cal H} $-measurable random variable.

Set $f(\omega)=\mathrm{ess.}\sup_{\xi\in \Xi^d, |\xi|\leq  K} F(\omega,\xi(\omega))$.  Then,
for almost all $\omega$,
\begin{eqnarray}
\label{lavieestbienfaite2}
f(\omega) & = & \sup_{y \in \mathbb{R}^d, |y| \leq K(\omega) } F(\omega,y).
\end{eqnarray}
\hfill $\square$
\end{lemma}


A compactness result involving random subsequences comes next, this is Lemma 2 of \cite{ks}.

\begin{lemma}\label{randomsub} Let $\mathcal{H}\subset\mathcal{F}$ be a sigma-algebra. 
Let $X_n$ be a sequence of $\mathcal{H}$-measurable $d$-dimensional random variables such that
\begin{equation}\label{raffaello}
\liminf_{n\to\infty}|X_n|<\infty
\end{equation}
almost surely. Then there exist
$\mathcal{H}$-measurable random variables 
$n_k:\Omega\to\mathbb{N}$, $k\in\mathbb{N}$ with $n_k(\omega)<n_{k+1}(\omega)$ for all $\omega\in\Omega$
and $k\in\mathbb{N}$ and an $\mathcal{H}$-measurable 
random variable $X$ such that $X_{n_k}\to X$ a.s. In such a case we write that 
\emph{there exists an $\mathcal{H}$-measurable
random subsequence}. \hfill $\square$
\end{lemma}

The following lemma uses the same idea as Lemma A.3 of \cite{lukasz}; it provides continuous versions
for certain random fields.

\begin{lemma}\label{szeder} Let $\mathcal{H}\subset\mathcal{F}$ be a sigma-algebra.
Define $\mathcal{K}:=[-N,N]^{n}$.
Let $L:\Omega\times \mathcal{K}\to \mathbb{R}$ be such that for a.e. $\omega\in\Omega$,
$L(\omega,\cdot)$ is continuous and for all $x\in \mathcal{K}$, $L(\cdot,x)$ is
measurable such that $\sup_{z\in \mathcal{K}} |L(\omega,z)|$ is integrable.
Then there is $l:\Omega\times \mathcal{K}\to \mathbb{R}$ such that for a.e. $\omega\in\Omega$,
$l(\omega,\cdot)$ is continuous and for all $k\in \mathcal{K}$, 
\begin{equation}\label{ekko}
E(L(k)|\mathcal{H})=l(k)\mbox{ a.s.}
\end{equation}
\end{lemma}
\begin{proof} $L$ is continuous on a full measure set $\tilde{\Omega}$. Let us redefine $L:=0$ outside
$\tilde{\Omega}$. It suffices to show \eqref{ekko} for this re-defined $L$. (Remember that all
sigma-algebras are assumed complete.)

Consider the separable Banach space 
$\mathbb{B}:=C([-N,N]^{n})$ of continuous functions on $[-N,N]^{n}$ with the supremum norm. Clearly, $L:\Omega\to \mathbb{B}$ and
for all $\mu$ in the dual space $\mathbb{B}'$ (which can be represented as a Borel signed measure), 
$\mu(L)=\int_{\mathcal{K}} L(\omega,x)\mu(dx)$
is a measurable function on $(\Omega,\mathcal{F})$: indeed, this is clear for
$\mu$ with finite support and then follows for general $\mu$ by approximation.
Note that, for each $k\in \mathcal{K}$, the linear functional
$f_k(x):=x(k)$, $x\in \mathbb{B}$ is continuous (w.r.t. the norm of $\mathbb{B}$) so
$f_k\in\mathbb{B}'$. Now it follows
from Proposition V.2.5. of \cite{neveu} that there is a measurable $l:\Omega\to \mathbb{B}$
such that $$
l(k)=f_k(l)=E(f_k(L)|\mathcal{H})=E(L(k)|\mathcal{H}),
$$ 
for each $k\in \mathcal{K}$, as claimed. 
\end{proof}

\section{Single Step Case}\label{single}

Now we turn to a set of lemmas that ensure that we can perform a backward iteration, 
and produce a series of actions that forms a candidate strategy of optimal execution.

Let $t\geq 1$ be an integer, let $\mathcal{G}, \mathcal{H}$ be ($\mathbb{P}$-complete) sigma-algebras 
over $\Omega$ such that $\mathcal{H} \subset \mathcal{G} \subset \mathcal{F}$ holds, and denote the set of $\mathcal{H}$-measurable 
$\mathbb{R}$-valued random variables by $\Xi$. We will consider functions
$$G_0:\Omega\times \mathbb{R} \times \mathbb{R}^{t}  \to \mathbb{R}, \ (\omega,x,v) \mapsto G_0(\omega,x,v),$$
and
$$V:\Omega\times \mathbb{R} \times \mathbb{R}^{t-1} \times \mathbb{R} \to \mathbb{R}, \ (\omega,x,j,h) \mapsto G_0(\omega,{}
x + \kappa_{t}(\omega,(j,h)),(j,h)),$$
where $(j,h) = (j_1,\hdots,j_{t-1},h)$, and $\kappa_t$ is as in (\ref{kappafg}).

Below we introduce the assumptions that will serve as a basis 
for the iterative arguments later.

\begin{assumption}\label{A1}
The function $G_0$ is 
$\mathcal{G}\otimes\mathcal{B}(\mathbb{R})\otimes\mathcal{B}(\mathbb{R}^{t})$-measurable.
For $\mathbb{P}$-almost every $\omega\in\Omega$,
the mapping $G_0(\omega,\cdot, \cdot)$, is jointly continuous 
and non-decreasing in its first variable. 
\end{assumption}

\begin{assumption}\label{A2}
There exists a measurable function $\bar G_0:\Omega\times\mathbb{R} \mapsto \mathbb{R}$,
such that, for $\mathbb{P}$-almost every $\omega\in\Omega$, 
$$\bar G_0(\omega,x) \to - \infty$$
holds as $x \to -\infty$; $\bar{G}_0(\omega,\cdot)$ is continuous, non-decreasing, 
and for all pairs $(x, v) \in \mathbb{R} \times \mathbb{R}^{t}$ 
\begin{equation*}
G_0(\omega,x,v) \leq \bar G_0(\omega, x) \leq C
\end{equation*}
hold for a fixed constant $C>0$.
\end{assumption}


\begin{assumption}\label{A3}
For any $m \in \mathbb{N}$ there exists an integrable $\mathcal{G}$-measurable random variable $M = M(m)$ such that,
outside a fixed $\mathbb{P}$-zero set,
for every $x \in [-m,m]$, $v \in [-m,m]^{t}$ and $1\leq t^{\star}\leq t$,
$$M(m)(\omega) \leq G_0(\omega,x + \kappa_{t^{\star}}(\omega,v_1,\hdots, v_{t^{\star}}),v)$$
holds.
\end{assumption}



Throughout Section \ref{single} we postulate that the conditions 
prescribed by Assumptions \ref{A1}, \ref{A2}, and \ref{A3} hold.
We remark that we will often
drop $\omega$ in the notation in the sequel and write only $G_{0}(x,v)$, $V(x,j,h)$, $\kappa_{t}(j,h)$; and similarly
for all other functions.

\begin{lemma}\label{continuous_version}
There exists $L : \Omega\times \mathbb{R} \times \mathbb{R}^{t-1} \times \mathbb{R} \to \mathbb{R}$ so that,
outside a fixed $\mathbb{P}$-zero set of $\omega\in\Omega$, $L(\omega,\cdot,\cdot,\cdot)$ 
is continuous in all its variables,
for all $(x,j,h) \in \mathbb{R} \times \mathbb{R}^{t-1} \times \mathbb{R}$ we have 
$$ E[V(x,j,h) | \mathcal{H}] = L(x,j,h).$$
Furthermore, if $X,J,H$ are $\mathcal{H}$-measurable random variables then a.s.
\begin{equation}\label{berak}
L(X,J,H)=E[V(X,J,H) | \mathcal{H}].
\end{equation}
\end{lemma}
 
\begin{proof}
Under Assumptions \ref{A1}, \ref{A2}, and \ref{A3} the proof is a direct consequence of Lemma \ref{szeder}.
Indeed, let $m\in\mathbb{N}$ be fixed for a moment. Outside a $\mathbb{P}$-zero set $N_{m}$, for all $(x,j,h)\in [-m,m]^{t+1}$ 
$$
M(m)\leq V(x,j,h)\leq C
$$
hold, hence Lemma \ref{szeder} implies the existence of $L^{m}$ satisfying the conclusion of this lemma for
$(x,j,h)\in [-m,m]^{t+1}$. 

Doing this for each $m$, by continuity, clearly $L^{m}=L^{m+1}$ on $[-m,m]^{t+1}$ outside a 
$\mathbb{P}$-zero set $N_{m}'\supset N_{m}$, so
we may unambiguously define $L(x,j,h):=L^{m}(x,j,h)$, $(x,j,h)\in [-m,m]^{t+1}$ outside $\cup_{m}N_{m}'$.
We set $L=0$ on $\cup_{m}N_{m}'$ and conclude.

Now we show \eqref{berak}. It is clear for step functions $X,J,H$.
We may and will assume $\max\{|X|,|J|,|H|\}\leq m$ for some $m>0$ and let us take step functions $X_{n}\to X$, $J_{n}\to J$, $H_{n}\to H$
with $\sup_{n}\max\{|X_{n}|,|J_{n}|,|H_{n}|\}\leq m$. By continuity, $L(X_{n},J_{n},H_{n})\to L(X,J,H)$ a.s.,
on the other hand, $$
E[V(X_{n},J_{n},H_{n}) | \mathcal{H}]\to E[V(X,J,H) | \mathcal{H}]\mbox{ a.s.}
$$ by the conditional dominated convergence theorem.
That theorem applies since $M(m)\leq V(X_{n},J_{n},H_{n})\leq C$. Finally, the events
$\{\max\{|X|,|J|,|H|\}\leq m\}$, $m\in\mathbb{N}$ exhaust $\Omega$ so the statement
is true also for unbounded $X,J,H$. The lemma is shown.
\end{proof}

\begin{lemma}\label{lemma::infinite_loss}
Let $L$ be as in Lemma \ref{continuous_version}. As $|h| \to \infty$ we have that,
outside a fixed $\mathbb{P}$-zero set,
\begin{align}\label{ruin}
\sup_{j \in \mathbb{R}^{t-1}}L(x,j,h) \to -\infty
\end{align}
for all $x$. 
(In the case $t=1$ we mean that $L$ does not depend on $j$ and there is no supremum.)
\end{lemma}
\begin{proof}
By Assumption \ref{A2}, the definition of $L$ and (\ref{lambda}), we have for all 
$(x,j,h) \in \mathbb{R}\times\mathbb{R}^{t-1}\times\mathbb{R}$ 
\begin{align*}
L(x,j,h) &= E[V(x,j,h) | \mathcal{H}]
\\ 
&\leq E[\bar G_0(x + \kappa_{t}(j,h))| \mathcal{H}]
\\
&
\leq 
E[\bar G_0(x + \lambda_{t}(h))| \mathcal{H}]
\end{align*}
almost surely. Take supremum in $j$ on the left-hand side and apply Fatou's reverse lemma on the right-hand side. 

By the absence of dependence on the variable $j$, by $\lambda_{t}(h)\to -\infty$, $|h|\to\infty$, and
by $\bar{G}_{0}(x)\to{}-\infty$, $x\to -\infty$, 
(\ref{ruin}) follows.
\end{proof}

\begin{lemma}\label{lemma::monotonicity_}
Outside a fixed $\mathbb{P}$-zero set, the inequality
\begin{align}\label{pook}
\begin{split}
&L(x_1, j, h) \leq L(x_2, j, h),
\end{split}
\end{align}
holds for all $x_1,x_2, h \in  \mathbb{R}$ with $x_1\leq x_2$, for all $j \in \mathbb{R}^{t-1}$.
\end{lemma}

\begin{proof} Let $\Omega'$ be a $\mathbb{P}$-full measure set such that \eqref{pook} holds for all $x_{1},x_{2},h\in\mathbb{Q}$,
$j\in\mathbb{Q}^{t-1}$
on $\Omega'$. Let $\Omega''$ be the full measure set on which $L$ is continuous. Then on $\Omega'\cap\Omega''$ \eqref{pook} holds
for \emph{all} $x_{1},x_{2},j,h$, by continuity.	

\end{proof}

\begin{lemma}\label{lemma::kettes}
Fix $l \in \mathbb{Z}$ and $m \in \mathbb{N}$. There exists an $\mathcal{H}$-measurable $K = K(l,m)$ such that,
outside a fixed $\mathbb{P}$-zero set, 
\begin{equation}\label{majus6}
L(x,j, h) \leq L(x, j, h \indicator{|h|\leq K}),
\end{equation}
for all $x \in [l,l+1]$, $j \in [-m,m]^{t-1}$, $h \in \mathbb{R}$.
\end{lemma}

\begin{proof}
Let $\bar \Omega$ denote the full measure set 
where the conclusions of Lemma \ref{continuous_version} and Lemma \ref{lemma::monotonicity_} hold, as well as 
\eqref{ruin} with $x=l+1$. 
By the measurable selection theorem (see e.g. Proposition III.44 of \cite{dm}) and by continuity of $L$, 
there is a random variable $M^{-} = M^{-}(l,m)$ such that, on the full measure set $\bar{\Omega}$,  
\begin{align}\label{allin}
L(\omega,l, M^{-}, 0)\leq 
\inf_{j\in [-m,m]^{t-1}}L(\omega,l, j, 0).
\end{align}
Indeed, the set $$
\mathcal{Z}:=\{(\omega,z)\in\Omega\times\mathbb{R}^{t-1}:L(\omega,l,z,0)\leq \inf_{j\in [-m,m]^{t-1}\cap\mathbb{Q}^{t-1}}
L(\omega,l, j, 0)\}
$$ is $\mathcal{H}\otimes \mathcal{B}(\mathbb{R}^{t-1})$-measurable and 
has nonempty $\omega$-sections for all $\omega\in\bar{\Omega}$ (by compactness of $[-m,m]^{t-1}$ and by continuity of $L$,
continuity of $L$ guarantees that taking an infimum
over rationals or reals is the same).
Then $$
\mathcal{Z}':=\mathcal{Z}\cup [(\Omega\setminus\bar{\Omega})\times\mathbb{R}^{t-1}]
$$ has nonempty $\omega$-sections{}
for all $\omega\in\Omega$ and it is still $\mathcal{H}\otimes \mathcal{B}(\mathbb{R}^{t-1})$-measurable.
Hence an $\mathcal{H}$-measurable selector $M^{-}$ for $\mathcal{Z}'$ exists which satisfies \eqref{allin}
on $\bar{\Omega}$. 

Fix $l \in \mathbb{Z}$ and $m \in \mathbb{N}$. Using Lemma \ref{lemma::infinite_loss} for all $\omega \in \bar\Omega$ 
there exists a $\mathcal{H}$-measurable random variable $K(\omega) = K(\omega, l, m)$ so that for all $h \in \mathbb{R}$ it holds that
\begin{align}\label{M}
\begin{split}
|h| > K(\omega) \implies \sup_{j\in [-m,m]^{t-1}}L(\omega,l+1, j, h)\leq L(\omega,l, M^{-}, 0).
\end{split}
\end{align}
Indeed, by continuity of $L$, 
$$
\sup_{j\in [-m,m]^{t-1}}L(\omega,l+1, j, h)\leq \sup_{j\in \mathbb{R}^{t-1}}L(\omega,l+1, j, h)={}
\sup_{j\in \mathbb{Q}^{t-1}}L(\omega,l+1, j, h) 
$$
outside a fixed $\mathbb{P}$-zero set, and the right-hand side is $\mathcal{H}$-measurable and converges to
$-\infty$ as $|h|\to\infty$. Define 
$$
K:=\inf\{n\in\mathbb{N}: \sup_{h\in\mathbb{Q},|h|> n} \sup_{j\in \mathbb{Q}^{t-1}}L(l+1, j, h)\leq 
L(l, M^{-}, 0)\},
$$
this will satisfy \eqref{M}.

On the event $\{|h| > K\} \cap \bar \Omega$, using Lemma \ref{lemma::monotonicity_}, 
and the statement in (\ref{M}), we have 
\begin{align}\label{boundidstratineq}
\begin{split}
&L(x, j, h) \leq L(l+1, j, h) \leq 
\sup_{j\in [-m,m]^{t-1}}L(l+1, j, h),
\\
&\leq L(l, M^{-}, 0) \leq  L(l, j, 0) \leq 
L(x, j, 0),
\end{split}
\end{align}
for every $x \in [l,l+1]$, $j \in [-m,m]^{t-1}$, $h\in \mathbb{R}$, and for all $\omega \in \bar \Omega$, 
completing the argument.
\end{proof}

\begin{lemma}\label{lemma::pointwise_essup}
There exists an $\mathcal{H}\otimes\mathcal{B}(\mathbb{R}^{t})$-measurable function 
$$G :\Omega\times\mathbb{R}\times \mathbb{R}^{t-1} \mapsto \mathbb{R},$$ such that, outside
a fixed $\mathbb{P}$-zero set, $G$ is 
continuous, non-decreasing in its first variable, furthermore, for all $x$, and for all $j \in \mathbb{R}^{t-1}$ we have 
\begin{align}\label{majus5}
\begin{split}
G(x, j) =  \mathrm{ess.}\sup_{H \in \Xi} L(x, j, H),
\end{split}
\end{align}
that is, $G(x,j)$ is a version of the essential supremum.
\end{lemma}

\begin{proof} We follow arguments of Lemma 6 in \cite{laurence}. For simplicity, assume $t=2$.
Let $x,j \in \mathbb{R}$, and set $l\in\mathbb{Z}$ 
and $m \in \mathbb{N}$ so that $x \in [l,l+1]$ and $j \in [-m,m]$ hold. We will 
prove the statement working elementwise on $\Omega$ so -- unlike before -- we
do not omit displaying dependence on $\omega\in\Omega$ throughout the proof, until further notice. 

Denote by $\bar \Omega$ the full measure set for which the conclusions of Lemmas \ref{continuous_version}, 
\ref{lemma::monotonicity_}, \ref{lemma::kettes} hold, and let $\omega \in \bar \Omega$. Let us define  
\begin{align*}
 B(\omega, x, j) = \sup_{h \in \mathbb{Q}, |h| \leq K(\omega)}L(\omega, x,j,h) = \sup_{h \in \mathbb{R}, |h| \leq K(\omega)}L(\omega, x,j,h),
\end{align*}
where $K(\omega) = K(\omega, l,m) = K(l,m)$ is as in Lemma \ref{lemma::kettes}. 
This is measurable, being the supremum of countably many measurable functions. Let us fix furthermore a sequence 
$(x_n,j_n)_{n \in \mathbb{N}} \subset [l,l+1]\times[-m,m]$ for which $(x_n,j_n) \to (x,j)$ as $n\to\infty$. 
Observe, that by definition of $B$, for every $k \in \mathbb{N}$ there exists $h_k(\omega, x,j)$, with 
$|h_k(\omega, x,j)| \leq K(\omega)$, so that $B(\omega, x,j) -1/k \leq L(\omega, x,j,h_k(\omega, x,j))$. 
The fact that for all $n\in\mathbb{N}$ we have $B(\omega,x_n,j_n) \geq L(\omega,x_n,j_n,h_k(\omega, x,j))$, along with continuity of $L$ yields
\begin{align*}
\liminf_{n}B(\omega,x_n,j_n) \geq & L(\omega,x,j,h_k(\omega, x,j)) 
\\
&
\geq
B(\omega, x,j) -1/k,
\end{align*}
which in the limiting case of $k \to \infty$ means $\liminf_{n}B(\omega,x_n,j_n) \geq B(\omega, x,j)$. 

Take a sequence $\{n_k, \ k\in\mathbb{N}\} \subset \mathbb{N}$ so that
\begin{align}\label{moha}
\limsup_{n}B(\omega,x_n,j_n) = \lim_{k}B(\omega,x_{n_k},j_{n_k}).
\end{align}
Since $\{h : h\in\mathbb{Q}, \ |h|\leq K(\omega) \}$ is a precompact set in $\mathbb{R}$, for every 
$k\in\mathbb{N}$ there exists $h_{n_k}^{*}\in\mathbb{R}$ with $|h_{n_k}^{*}(\omega)| \leq K(\omega)$ so that 
$B(\omega,x_{n_k},j_{n_k}) = L(\omega,x_{n_k},j_{n_k},h_{n_k}^{*}(\omega))$. Using  the compactness of the 
closure there exists $h^*(\omega)\in\mathbb{R}$, $|h^*(\omega)|\leq K(\omega)$ and a subsequence 
$(a_k)_{k\in\mathbb{N}}$ of $\{n_k : k\in\mathbb{N}\}$ so 
that $h^{*}_{a_k}(\omega) \to h^*(\omega)$, $k\to\infty$. These, and (\ref{moha}) imply
\begin{align*}
\limsup_{n} & B(\omega,x_n,j_n) = \lim_{k}B(\omega,x_{n_k},j_{n_k}) = \lim_{k}B(\omega,x_{a_k},j_{a_k})
\\
&
= \lim_{k}L(\omega,x_{a_k},j_{a_k},h_{a_k}^{*}(\omega)) = L(\omega,x,j,h^{*}(\omega)) \leq B(\omega,x,j),
\end{align*}
establishing the continuity of $B$. 

As far as monotonicity is concerned, the mapping $x \to B(\omega,x,j)$, $x\in\mathbb{R}$ inherits the 
non-decreasing property from $L$ (stated in Lemma \ref{lemma::monotonicity_}) naturally.

From Lemma \ref{gast} it follows that
\begin{align*}
B(\omega,x,j) = \mbox{ess.sup}_{H\in\Xi, |H(\omega)| \leq K(\omega)} L(\omega, x,j, H(\omega)).
\end{align*}
In the discussion above $\omega \in \bar \Omega$ was arbitrary, and returning to the usual practice of not displaying the dependence on it, Lemma \ref{lemma::kettes} shows that 
\begin{align*}
\mbox{ess.sup}_{H \in \Xi} L( x,j, H)  \leq \mbox{ess.sup}_{H\in\Xi, |H| \leq K}& L(x,j, H) = B(x,j) 
 \\
 &\leq
 \mbox{ess.sup}_{H \in \Xi} L( x,j, H)
\end{align*}
holds almost surely, so we can set $G:=B$.

\end{proof}

\begin{lemma}\label{lemma::harmas}
Let $X,H_1,\hdots, H_{t-1} $ be $\mathcal{H}$-measurable random variables, 
and let $G$ be as in Lemma \ref{lemma::pointwise_essup}. Then, the quantity $G(X,(H_1,\hdots, H_{t-1}))$ is a 
version of the essential supremum
\begin{align*}
\begin{split}
\mathrm{ess.}\sup_{H\in \Xi}
L(X,(H_1,\hdots,H_{t-1}), H).
\end{split}
\end{align*}
\end{lemma}

\begin{proof}
Without loss of generality we assume that $t = 2$, and that $X$ and $H_1$ take values in $[l,l+1]$ and $[-m,m]$ 
respectively. Let $X_n$ and $H_1^{(n)}$ be $\mathcal{H}$-measurable random variables, taking values in $[l,l+1]\cap\mathbb{Q}$ and 
$[-m,m]\cap\mathbb{Q}$, respectively, for all $n \in \mathbb{N}$, and possessing also 
the limiting properties $X_n \to X$ and $H_1^{(n)} \to H_1$.

Observe that for all $n\in\mathbb{N}$, on a full measure set, we have that 
$$\mathrm{ess.sup}_{H\in \Xi} L(X_n,H_1^{(n)},H) = G(X_n,H_1^{(n)}),$$ and furthermore, 
as a consequence, for all $n\in\mathbb{N}$ there exists $H_n\in\Xi$ so that 
\begin{align}\label{fin1}
 L(X_n,H_1^{(n)},H_n) \geq  G(X_n,H_1^{(n)}) - 1/n
\end{align}
almost surely. According to Lemma \ref{lemma::kettes} there exists an $\mathcal{H}$-measurable $K = K(l,m) = K(\omega,l,m)$ so that
\begin{align}\label{fin2}
 L(X_n,H_1^{(n)},H_n \indicator{|H_n|\leq K}) \geq  L(X_n,H_1^{(n)},H_n).
\end{align}
Note that the $K$ does not depend on the integer $n$ in any way. Putting together (\ref{fin1}) and (\ref{fin2}) gives, for all $n \in \mathbb{N}$ the almost sure inequality
\begin{align}\label{fin3}
L(X_n,H_1^{(n)},H_n \indicator{|H_n|\leq K}) \geq  G(X_n,H_1^{(n)}) - 1/n.
\end{align}
Now using Lemma \ref{randomsub}, there exists an $\mathcal{H}$-measurable subsequence $k_n = k_n(\omega), \ n\in\mathbb{N}$ and an $\mathcal{H}$-measurable $\bar H$ so that $H_{k_n} \indicator{|H_{k_n}|\leq K} \to \bar H$ holds almost surely. In (\ref{fin3}) taking the limit as $n \to \infty$ along the sequence $k_n, \ n\in\mathbb{N}$, and utilizing continuity of $L$ and $G$ yields
\begin{align*}
L(X,H_1,\bar H) \geq G(X,H_1),
\end{align*}
which in return implies $\mathrm{ess.sup}_{H\in \Xi}L(X,H_1,H) \geq G(X,H_1)$.

On the other hand, by definition of $G$, for every $\mathcal{H}$-measurable $H$ we have 
\begin{align*}
L(X,H_1,H) =& \lim_{n}L(X_n,H_1^{(n)},H) \leq \lim_{n} \mbox{ess.sup}_{H}L(X_n,H_1^{(n)},H)
\\
&=
 \lim_{n} G(X_n, H_1^{(n)}) = G(X,H_1).
\end{align*}
Taking the essential supremum of both sides yields the inequality
\begin{align*}
\mathrm{ess.sup}_{H\in \Xi}L(X,H_1,H) \leq G(X, H_1)
\end{align*}
on a full measure set, finishing the proof.
\end{proof}

\begin{lemma}\label{optimal_H}
Let $G$ be as in Lemma \ref{lemma::pointwise_essup}. 
Let $X, H_1,\hdots, H_{t-1}$ be $\mathcal{H}$-measurable random variables. Then
there exists an $\mathcal{H}$-measurable $H^*$ such that
\begin{align*}
G(X,(H_1,\hdots,H_{t-1})) = L(X,(H_1,\hdots, H_{t-1}), H^*)
\end{align*}
holds almost surely.
\end{lemma}

\begin{proof}
Without loss of generality assume that $t=2$.
We may also assume that $X$ almost surely takes values in the closed interval 
$[l,l+1]$, $H_1$ takes values in the closed interval $[-m,m]$ for some $l \in \mathbb{Z}$ and for some $m \in \mathbb{N}$. Let us define
\begin{equation*}
\mathcal{O} = \Big\{  L(X,H_1, H), \ \ H \in \Xi   \Big\}.
\end{equation*}
This set is upward directed in terms of the almost sure sense of the usual ``less than or equal'' relation. 
Indeed, for $H^{\dagger},H^{\ddagger}\in \Xi$, we have
$$
\max\{L(X,H_{1},H^{\dagger}),L(X,H_{1},H^{\ddagger})\}\leq L(X,H_{1},H^{\sharp})\mbox{ a.s.},
$$
where
$$
H^{\sharp}:=H^{\dagger}\mathbf{1}_{\{L(X,H_{1},H^{\dagger})>L(X,H_{1},H^{\ddagger})\}}+
H^{\ddagger}\mathbf{1}_{\{L(X,H_{1},H^{\dagger})\leq L(X,H_{1},H^{\ddagger})\}}\in\Xi.
$$

Thus, using Proposition VI.1.1. of \cite{neveu} we have that there exists a sequence 
$\{ H_n, \ n\in\mathbb{N} \} \subset \Xi$ for which the limiting property
\begin{align}\label{krokodil0}
\begin{split}
 L(X, H_1, H_n ) \to G(X,H_1) 
\end{split}
\end{align}
holds in a non-decreasing way almost surely as $n \to \infty$. Utilizing Lemma \ref{lemma::kettes}, and the $K = K(l,m)$ within, we have for each $n \in\mathbb{N}$ that
\begin{align}\label{beka0}
\begin{split}
L&(  X, H_1, H_n) \leq L(X, H_1, H_n \indicator{|H_n| \leq K}),
\end{split}
\end{align}
almost surely. Then, again using Lemma \ref{randomsub}, there exists an $\mathcal{H}$-measurable subsequence, 
say $k_n, \ n \in \mathbb{N}$, and there exists an $\mathcal{H}$-measurable $H^*$ so that 
$$
H_{k_n} \indicator{|H_{k_n}| \leq K} \to H^*
$$
in the almost sure sense. Continuity of $L$, (\ref{krokodil0}), and (\ref{beka0}) together guarantees
\begin{align*}
\begin{split}
G(X,H_1) &= \lim_{n \to \infty} L(X,H_1, H_{k_n} ) 
\\
\leq&
\lim_{n \to \infty} L(X, H_{1},H_{k_n} \indicator{|H_{k_n}| \leq K}) = L(X,H_1, H^*).
\end{split}
\end{align*}
The proof is complete.
\end{proof}

\begin{lemma}\label{x_bound}
Let $G$ be as in Lemma \ref{lemma::pointwise_essup}. There exists a measurable
$\bar G :\Omega\times \mathbb{R} \to \mathbb{R}$ such that, outside a fixed $\mathbb{P}$-zero set, the 
following requirements are met: $\bar{G}$ is nondecreasing and continuous, as $x \to -\infty$ we have
$$\bar G(x) \to - \infty,$$
and, for all $x \in \mathbb{R}$, $j \in \mathbb{R}^{t-1}$ we have 
\begin{equation*}
G(x,j) \leq \bar G(x) \leq C.
\end{equation*}
\end{lemma}
\begin{proof}
Let $G$ be as in Lemma \ref{lemma::pointwise_essup} and let $H^{*}$ be as in Lemma \ref{optimal_H},
corresponding to $X=x$, $H_{1}=j_{1}$, \ldots, $H_{t}=j_{t-1}$. 
We have for $x\in \mathbb{R}$, $j\in\mathbb{R}^{t-1}$ that
\begin{align*}
G(x,j) =& L(x,j,H^*) =  {E}[V(x,j,H^{*})|\mathcal{H}] 
\\
&= {E}[G_0(x +\kappa_t(j, H^*),(j,H^{*}))|\mathcal{H}]  \leq {E}[\bar G_0(x +\kappa_t(j, H^*))|\mathcal{H}]
\\
&
\leq {E}[\bar G_0(x + P_t^2 \delta_t / 4)| \mathcal{H}]
\end{align*}
almost surely. 
For $x\in [-m,m]$, 
$$
M(m)\leq \bar{G}_{0}(-m)\leq \bar{G}_{0}(x+P_{t}^{2}\delta_{t}/4)\leq C,
$$
hence, the argument of Lemma \ref{continuous_version} provides a 
version of $E[\bar G_0(x + P_t^2 \delta_t / 4)| \mathcal{H}]$ which,
outside a fixed $\mathbb{P}$-zero set, is continuous in $x$. We denote this version $\bar{G}(x)$.
 
Arguing along rationals and using continuity, this version is non-decreasing as $\bar{G}_{0}$ is, 
outside a fixed $\mathbb{P}$-zero set.
Finally, if $x\to -\infty$ then also $\bar{G}_{0}(x + P_t^2 \delta_t / 4)\to -\infty$ almost surely and the reverse
Fatou lemma shows $\bar{G}(x)\to -\infty$ almost surely.

\end{proof}

\begin{lemma}\label{lemma::lower_bd}
Let $G$ be as in Lemma \ref{lemma::pointwise_essup}. For any $m\in \mathbb{N}$ 
there exists an $\mathcal{H}$-measurable 
integrable random variable $ M_1 = M_1(m)$ so that $$M_1 \leq G(x + \kappa_{t^{\star}}(j_1,\hdots, j_{t^{\star}}),j)$$ for every 
$x \in [-m,m]$, $j \in [-m,m]^{t-1}$, $t^{\star}\in \{1,\ldots,t-1\}$, outside a fixed $\mathbb{P}$-zero set.
\end{lemma}
\begin{proof}
Let $m \in \mathbb{N}$ and $t^{\star} \in \{1,\hdots,t-1\} $. Using Assumption \ref{A3} 
there exists an integrable, $\mathcal{G}$-measurable $M = M(m)$ so that 
with the notation $(j,0) = (j_1,\hdots, j_{t-1}, 0)$ we have 
\begin{align*}
G_0(x + \kappa_{t^{\star}}(j_1,\hdots,j_{t^{\star}}), (j, 0)) \geq M
\end{align*}
for all $x \in [-m,m]$, for all $j \in [-m,m]^{t-1}$, outside a fixed $\mathbb{P}$-zero set. 
For fixed $j = (j_1, \hdots, j_{t-1}) \in [-m,m]^{t-1}$ and $x \in [-m,m]$, we have  
\begin{align}\label{hamv}
\begin{split}
G(x + &\kappa_{t^{\star}}(j_1,\hdots,j_{t^{\star}}),j) \geq L(x + \kappa_{t^{\star}}(j_1,\hdots,j_{t^{\star}}),j,0)
\\
&
= E[V(x + \kappa_{t^{\star}}(j_1,\hdots,j_{t^{\star}}), j, 0)| \mathcal{H}] 
\\
&
= E[G_0(x + \kappa_{t^{\star}}(j_1,\hdots,j_{t^{\star}}) + \kappa_{t}(j, 0), (j, 0))| \mathcal{H}]
\\
&
= E[G_0(x + \kappa_{t^{\star}}(j_1,\hdots,j_{t^{\star}}), (j, 0))| \mathcal{H}] \geq  E[M| \mathcal{H}]
\end{split}
\end{align}
almost surely. Since $G,\kappa$ are continuous the relation established in (\ref{hamv}) above holds for all $x \in [-m,m]$, for 
all $j \in [-m,m]^{t-1}$, outside a fixed $\mathbb{P}$-zero set. 
The choice $M_1 = E[M| \mathcal{H}]$ gives a desired lower bound. 

\end{proof}

\section{The Generic Step, Dynamic Programming}\label{dyna}

The proof of Theorem \ref{main} will consist of two major parts.
First, in a phase of \emph{backward induction} we construct the optimal value functions 
at times $T,T-1,\ldots,0$ and prove their regularity (monotonicity, continuity), using the lemmas of Section \ref{single}. 

Second, with a \emph{forward iteration} we construct a \emph{candidate optimal strategy}, relying on repeated applications
of Lemma \ref{optimal_H}. Since the strategies of that lemma attain the respective essential
suprema over all possible strategies, we will be able to show that the candidate strategy thus constructed
indeed represents a series of actions that dominates all admissible executions in terms of expected utility.
In this way, we will arrive at the conclusion of the paper.



Assumptions \ref{uint} and \ref{deltacska} will be in force from now on. Fix $C_{u}\geq 0$ such that $u(x)\leq C_{u}$ for
all $x\in\mathbb{R}$. We denote by $\Xi_{t}$ the family of $\mathcal{F}_{t}$-measurable random variables.

\begin{proof}[Proof of Theorem \ref{main}]
Define $\tilde{G}_{T}:\Omega\times \mathbb{R}\times \mathbb{R}^{T}\to \mathbb{R}$ as
\begin{equation}\label{okt1}
\tilde{G}_{T}(\omega,x,h_{1},\ldots,h_{T}):=u(x-B(\omega)),\ \omega\in\Omega,\ x,h_{1},\ldots,h_{T}\in\mathbb{R}.
\end{equation}
Note that $h_{1},\ldots,h_{T}$ are dummy variables here, and $\tilde{G}_{T}$ is continuous and nondecreasing in $x$,
outside a fixed $\mathbb{P}$-zero set. 

The first step of the \emph{backward induction} is different from the other steps
since $h_{T}=-h_{1}-\ldots-h_{T-1}$ due to the constraint on liquidation. 
We thus consider $\tilde{G}_{T-1} : \Omega \times \mathbb{R} \times \mathbb{R}^{T-1} \to \mathbb{R}$ with the definition
\begin{align}\label{okt2}
\begin{split}
& \tilde{G}_{T-1}(\omega,x,h_{1},\ldots,h_{T-1}) \\
&=
E\left[\tilde{G}_{T}\left(x+\kappa_{T}\left(h_{1},\ldots,h_{T-1},-\sum_{k=1}^{T-1}h_{k}\right), h_1, \ldots, h_{T-1},
-\sum_{k=1}^{T-1}h_{k}\right)\left\vert\mathcal{F}_{T-1}\right.\right].
\end{split}
\end{align}

If $x,h_{1},\ldots,h_{T-1}\in [-m,m]$ then  
\begin{eqnarray*}
C_{u} &\geq& \tilde{G}_{T}\left(x+\kappa_{T}\left(h_{1},\ldots,h_{T-1},-\sum_{k=1}^{T-1}h_{k}\right), h_1, \ldots, h_{T-1},
-\sum_{k=1}^{T-1}h_{k}\right)\\
&\geq& \min\{u(-m-P_{T}(T-1)m-\zeta_{0}(T-1)m-2(T-1)^{2}m^{2}/\delta-B),\\
& & u(-m+P_{T}(T-1)m-\zeta_{0}(T-1)m-2(T-1)^{2}m^{2}/\delta-B) \},
\end{eqnarray*}
and both random variables on the right-hand side are integrable by Assumption \ref{uint}.
Now the argument of Lemma \ref{continuous_version} guarantees that there is a version of $\tilde{G}_{T-1}$ which, outside
a fixed $\mathbb{P}$-zero set, is continuous in all its variables and non-decreasing in $x$.

To start the \emph{backward induction} one has to examine whether the conditions prescribed by Assumption \ref{A1}, 
Assumption \ref{A2}, and Assumption \ref{A3} hold with the choice $G_0 = \tilde{G}_{T-1}$. 

We note that 
the function $\tilde{G}_{T-1}$ is jointly measurable, being jointly continuous in its real variables and being
measurable when fixing these. 

Using the bound in (\ref{lambda}) we define $\hat G_{T-1}: \Omega \times \mathbb{R} \to \mathbb{R}$ as
$$\hat G_{T-1}(\omega,x) = E\left[u\left(x + \frac{P_T^2(\omega) \delta_T(\omega)}{4}-B\right)\vert\mathcal{F}_{T-1}\right].$${}
Since
$$
-|u(-m-B)|\leq u\left(x + \frac{P_T^2(\omega) \delta_T(\omega)}{4}-B\right)\leq C_{u}
$$ 
and Assumption \ref{uint} holds, the argument of Lemma \ref{continuous_version}
shows that $\hat{G}_{T-1}$
has an a.s.\ continuous and non-decreasing version. 

We note that -- due to Assumption \ref{uint} -- for all 
$x \in \mathbb{R}$ it holds that  $\hat G_{T-1}(x) \leq C_{u}$, and as $x \to -\infty$ the quantity $\hat G_{T-1}(x)$ tends to 
$- \infty$ as $x\to -\infty$, in the $\mathbb{P}$-almost sure sense: 
$u$ does so by assumption and due to boundedness from above the reverse Fatou Lemma is applicable. 
Moreover, for all $x, h_1, \ldots, h_{T-1}$, we have $$\tilde{G}_{T-1}(\omega,x,h_{1},\ldots,h_{T-1}) \leq \hat{G}_{T-1}(\omega,x)$$ 
outside a fixed $\mathbb{P}$-zero set, since this holds along rationals and both sides are continuous.

We will establish the following claim after the present proof.

\begin{claim}\label{claim} For any $m > 0$  there exists 
an $\mathcal{F}_{T-1}$-measurable and integrable $M = M(m)$ such that for all $1\leq t\leq T-1$ and for all  
$(x,h_1, \ldots, h_{T-1}) \in [-m,m]^{T}$ 
we have $$M(m) \leq \tilde{G}_{T-1}(x + \kappa_{t}(h_1, \ldots, h_{t}), h_1, \ldots, h_{T-1})$$
almost surely.
\end{claim}

One can thus conclude that Assumption \ref{A1}, Assumption \ref{A2}, and Assumption \ref{A3} of Section \ref{single} 
are satisfied and we are ready to perform the first step of the \emph{backward induction}: the 
lemmas of Section \ref{single}, for the first step, will be utilized with the choice $t = T-1$ and $G_0 = \tilde{G}_{T-1}$,
$\bar{G}_{0}:=\hat{G}_{T-1}$.

Lemmas \ref{continuous_version}, \ref{lemma::infinite_loss}, \ref{lemma::monotonicity_}, \ref{lemma::kettes}, 
and \ref{lemma::pointwise_essup} produce a mapping -- denoted by $G$ in their own context -- 
which, in our notation will take the form $\tilde{G}_{T-2}:\Omega \times \mathbb{R}^{T-1} \to \mathbb{R}$ with the following properties. 
The $\tilde{G}_{T-2}$ is $\mathcal{F}_{T-2} \otimes \mathcal{B}(\mathbb{R}^{T-1})$-measurable, non-decreasing in its first real variable, 
jointly continuous in all its real variables, and for all $X, H_1, \ldots, H_{T-2} \in \Xi_{T-2}$ it almost surely satisfies
\begin{align}\label{okt3}
\begin{split}
& \tilde{G}_{T-2}(X, H_1, \ldots, H_{T-2})
\\
&=
\mbox{ess.sup}_{H \in \Xi_{T-2}}E\left[\tilde{G}_{T-1}(X + \kappa_{T-1}(H_1, \ldots, H_{T-2}, H), 
H_1, \ldots, H_{T-2}, H) \vert \mathcal{F}_{T-2} \right].
\end{split}
\end{align}


Lemmas \ref{x_bound}, and \ref{lemma::lower_bd} imply that there exists an action-independent bound 
$\hat G_{T-2}$, with properties analogous to the $\hat G_{T-1}$ presented above, and thus finally, 
we arrive to the conclusion that the quantity $\tilde{G}_{T-2}$ is such that it again satisfies 
Assumption \ref{A1}, Assumption \ref{A2}, and Assumption \ref{A3}. That is, in 
the next step of iteration, the choice $G_0 = \tilde{G}_{T-2}$, $\bar{G}_{0}:=\hat{G}_{T-2}$ can be made.
Iterating backwards in this manner goes with the usual mechanics of induction. 

When $\tilde{G}_{t^{\star}}, \ldots, \tilde{G}_{T}$ 
are given for some $t^{\star} \leq T-2$, applying the lemmas of Section  
\ref{single} with the choice $G_{0} = \tilde{G}_{t^{\star}}$, $\bar{G}_{0}:=\hat{G}_{t^{\star}}$ and $t = t^{\star}$ 
yield $\tilde{G}_{t^{\star}-1}$, $\hat{G}_{t^{\star}-1}$ and with this procedure we construct 
\begin{equation}\label{pairs}
\tilde{G}_{T-2}, \tilde{G}_{T-3}, \ldots, \tilde{G}_{0}, \hat{G}_{T-2},\ldots,\hat{G}_{0}
\end{equation}
with the properties listed below:


For $t \in \{0,1, \ldots, T-2\}$, $\tilde{G}_{t}:\Omega \times \mathbb{R}^{t+1} \to \mathbb{R}$ is 
$\mathcal{F}_{t} \otimes \mathcal{B}(\mathbb{R}^{t+1})$-measurable, non-decreasing in its first 
real variable, jointly continuous in its real variables, and for all $X, H_1, \ldots, H_{t} \in \Xi_{t}$, 
in the almost sure sense we have
\begin{align*}
\begin{split}
\tilde{G}_{t}(X, H_1,& \ldots, H_{t})
\\
=&
\mbox{ess.sup}_{H \in \Xi_{t}}
E\left[\tilde{G}_{t+1}(X + \kappa_{t+1}(H_1,\ldots,H_{t},H), H_1, \ldots, H_{t}, H) \vert \mathcal{F}_t \right].
\end{split}
\end{align*}
In the case $t=0$,
$$
\tilde{G}_{0}(x)=\mbox{ess.sup}_{H \in \Xi_{0}}
E\left[\tilde{G}_{1}(x + \kappa_{1}(H),H) \vert \mathcal{F}_0 \right].
$$
Clearly, $\Xi_{0}$ can be identified with $\mathbb{R}$ here, since $\mathcal{F}_{0}$ is the trivial sigma-algebra.



Now we construct the \emph{candidate strategy} using Lemma \ref{optimal_H}. In the case $t=1$,
with the choice $X=z$ we obtain some $H^{*}$ which we denote by $\hat H_1 $. 
We define the optimal steps using a forward recursion. That is, when the $\hat H_t$ is constructed for 
some $1\leq t \leq T-1$, with accumulated wealth $\hat X_{t} = \xi_{t}^{\hat H}$, the 
next action $\hat H_{t+1}$ is given by Lemma \ref{optimal_H} with the choice
$X=\hat{X}_{t}$, $H_{1}=\hat{H}_{1}$, \ldots, $H_{t}=\hat{H}_{t}$.
Finally, we set $\hat H_T = -\hat H_{T-1} - \ldots - \hat H_1$.


We proceed with the \emph{forward iteration}, giving an upper bound for all expected payoffs which -- 
as we shall see -- equals the expected payoff associated with the \emph{candidate strategy} $\hat H$.

Letting $(H_1, \ldots, H_{T-1}, H_T)\in\mathcal{A}_{0}$ denote an arbitrary series of admissible actions. The definition of the functions
$\tilde{G}_{t}$ yields
\begin{align}\label{arvaszamar}
\begin{split}
&E[u(z + \xi_T^{H} - B)]
\\
&=
E[ E[u(z + \xi_{T-1}^{H} + \kappa_T(H_1, \ldots, H_{T-1}, - \sum_{k = 1}^{T-1} H_{k}) - B)] \vert \mathcal{F}_{T-1}]]
\\
&=
E[ \tilde{G}_{T-1}(z + \xi_{T-2}^{H} + \kappa_{T-1}(H_1, \ldots, H_{T-1}), H_1, \ldots, H_{T-1}) ]
\\
& 
\leq
E[ E[ \tilde{G}_{T-2}(z + \xi_{T-3}^{H} + \kappa_{T-2}(H_1, \ldots, H_{T-2}), H_1, \ldots, H_{T-2}) \vert \mathcal{F}_{T-2} ] ]
\\
& \vdots
\\
&\phantom{} \leq  E[\tilde{G}_0(z)]=\tilde{G}_{0}(z).
\end{split}
\end{align}

Furthermore, as a result of (\ref{okt1}), (\ref{okt2}), (\ref{okt3}), and the construction of $\hat{H}$, we have that 
\begin{align}\label{csapzottmacska}
\begin{split}
\tilde{G}_0(z) =& E[ \tilde{G}_1(z + \kappa_1(\hat H_1), \hat H_1)]
\\
=&  E[\tilde{G}_2(z + \kappa_1(\hat H_1) + \kappa_2(\hat H_1, \hat H_2), \hat H_1, \hat H_2)]
\\
\vdots&
\\
=& E[\tilde{G}_{T-1}(z + \sum_{i = 1}^{T-1}\kappa_{i}(\hat H_1, \ldots, \hat H_{i}), \hat H_1, \ldots, \hat H_{T-1})]
\\
=& E[u(z + \xi_T^{\hat H} - B)].
\end{split}
\end{align}
Thus, we can conclude, due to (\ref{csapzottmacska}) and (\ref{arvaszamar}), that the candidate 
strategy given with trading actions $(\hat H_t)_{t = 1, \ldots, T}$, defined above, is indeed optimal and the indirect
utility $\bar{u}$ equals $\tilde{G}_{0}$. 
Theorem \ref{main} is now shown.
\end{proof}

\begin{proof}[Proof of Claim \ref{claim}]
For $h_{1},\ldots,h_{T-1}\in [-m,m]$ we clearly have
$$
\kappa_{t}(h_{1},\ldots,h_{t})\geq -m\zeta_{0}-tm^{2}/\delta+\min\{mP_{t},-mP_{t}\},
$$
for each $t=1,\ldots,T-1$.
Similarly,
\begin{eqnarray*} & & 
\kappa_{T}\left(h_{1},\ldots,h_{T-1},-\sum_{k=1}^{T-1}h_{k}\right)\\
&\geq& -(T-1)m\zeta_{0}-(T-1)(2T-2)m^{2}/\delta+\min\{m(T-1)P_{T},-m(T-1)P_{T}\}.
\end{eqnarray*}
Define the constant
$$
D_{m}:=-Tm\zeta_{0}-Tm^{2}/\delta-(T-1)(2T-2)m^{2}/\delta.
$$
We can thus see that
$$
\tilde{G}_{T-1}(x+\kappa_{t}(h_{1},\ldots,h_{t}),h_{1},\ldots,h_{T-1})\geq \mathbb{E}\left[{}
M(m)\left\vert\mathcal{F}_{T-1}\right.\right]
$$
where $M(m):=\min_{1\leq t\leq T-1}\min\{J_{1}(t),J_{2}(t),J_{3}(t),J_{4}(t)\}$ and
\begin{eqnarray*}
J_{1}(t) &=& u(-m+D_{m}+mP_{t}+m(T-1)P_{T}-B),\\
J_{2}(t) &=& u(-m+D_{m}-mP_{t}+m(T-1)P_{T}-B),\\
J_{3}(t) &=& u(-m+D_{m}+mP_{t}-m(T-1)P_{T}-B),\\
J_{4}(t) &=& u(-m+D_{m}-mP_{t}-m(T-1)P_{T}-B).	
\end{eqnarray*}
Now \eqref{inti} in Assumption \ref{uint} implies our statement.
\end{proof}

\section{Conclusion}\label{conclu}

One could add a solvency constraint (that is, $z+\xi^{X}_{t}\geq 0$ for all $t$) with minimal modifications
in the arguments. Setting $B:=0$, utility maximization for $u$ defined on the positive real axis 
could be treated in this way. A random utility $u$ could also be easily handled.

Extension to $u$ that are not bounded from above is also possible at the price
of complications that we wished to avoid in the current study. 
Considering several risky assets would lead to a considerably more
complex dynamics but the main line of arguments would prevail.

It is unclear whether (and how) results of the present paper could be transferred to continuous-time models.
We rely, in an essential way, on the fact that the treated portfolio optimization problem can be broken down into one-step
problems. In continuous time such an approach is out of question.

We do not assume concavity of $u$ so there is no hope for uniqueness of the optimal strategy
in general. Even if $u$ is assumed concave, the indirect utility may fail this property.
For this reason, a 
satisfactory duality theory seems out of reach. In the most favourable case, a suitable convexification of
the original problem might have the same value function as the original one: in this case
duality may come into play. Such investigations, however, are far beyond the scope of the present article.

\medskip{}

{\footnotesize \noindent\textbf{Acknowledgments.} Both authors gratefully 
acknowledge the support of 
the National Research, Development and Innovation Office (NKFIH) through grant K 143529 
and also within the framework of the Thematic Excellence Program 2021 (National Research subprogramme 
``Artificial intelligence, large networks, data security: mathematical foundation and applications'').
The second author also thanks for the support of NKFIH grant KKP 137490. We thank both Referees for
their comments which significantly improved the quality of our paper.}

\medskip{}

{\footnotesize \noindent\textbf{Data availability statement.} This research has generated no data.}


\end{document}